\documentclass[aps,amsmath,amssymb,twocolumn,nofootinbib,prx,superscriptaddress]{revtex4-2}

\usepackage{graphicx,tikz}
\usepackage{dcolumn}
\usepackage{bm}
\usepackage{epsfig,hyperref,amsthm}
\usepackage{mathtools}
\usepackage{color,fancybox}
\usepackage{colordvi}
\usepackage{relsize}
\usepackage{accents}
\usepackage{wasysym} 
\usepackage{enumitem}
\usepackage{mathtools,slashed}
\usepackage{times}
\usepackage{newtxmath}

\hypersetup{
	colorlinks=true,
	linkcolor=CitingColor,
	citecolor=CitingColor,
	urlcolor=CitingColor
}

\DeclareMathAlphabet\mathbfcal{OMS}{cmsy}{b}{n}
\newcommand{\ket}[1]{\ensuremath{|#1\rangle}}
\newcommand{\bra}[1]{\ensuremath{\langle #1|}}

\newcommand{\proj}[1]{\ket{#1}\bra{#1}}

\newcommand{\tr}{{\rm tr}}
\newcommand{\norm}[1]{\left\|#1\right\|}

\newcommand{\id}{\mathbb{I}}

\newcommand{\Perrorstate}{{\text{\bf P}_{\rm error}^{\rm st\text{-}exc}}}
\newcommand{\Perrorclass}{{\text{\bf P}_{\rm error}^{\rm cl\text{-}exc}}}
\newcommand{\Perrorens}{{\text{\bf P}_{\rm error}^{\rm ens\text{-}exc}}}
\newcommand{\perrorens}{{p_{\rm error}^{\rm ens\text{-}exc}}}
\newcommand{\cardin}{{n}}
\newcommand{\cardout}{{w}}
\newcommand{\stEns}{{\mathbfcal{S}}}
\newcommand{\prob}{{{\bm p}}}
\newcommand{\permu}{{{\bm\pi}}}
\newcommand{\averho}{{\rho^{\rm ave}_x}}
\newcommand{\Nset}{{\text{\bf N}}}
\newcommand{\PerrorstateTheta}{{\text{\bf P}_{\rm error}^{\rm st\text{-}exc (\Theta)}}}
\newcommand{\probunif}{{{\bm p}^{\rm uni}}}
\newcommand{\PerrorensTheta}{{\text{\bf P}_{\rm error}^{\rm ens\text{-}exc (\Theta)}}}
\newcommand{\perrorencrypt}{{p_{\rm error}^{\rm encrypt}}}
\newcommand{\Perrorencrypt}{{\text{\bf P}_{\rm error}^{\rm encrypt}}}

\newtheorem{result}{Theorem}
\newtheorem{alemma}[result]{Lemma}

\newtheorem{atheorem}[result]{Proposition}

\newtheorem{acorollary}[result]{Corollary}

\newtheorem{maincoro}[result]{Corollary}

\newtheorem{question}{Question}

\newcommand{\Q}[1]{{\begin{center}\color{purple}\fbox{\begin{minipage}{.95\columnwidth}\small {\color{purple}\begin{question}{\rm #1}\end{question}}\end{minipage}}\end{center}}}

\definecolor{nred}{rgb}{0.9,0.1,0.1}
\definecolor{nblack}{rgb}{0,0,0}
\definecolor{nblue}{rgb}{0.2,0.2,0.8}
\definecolor{ngreen}{rgb}{0.2,0.6,0.2}
\definecolor{ublue}{rgb}{0,0,0.5}

\definecolor{pur}{rgb}{0.75,0,0.75}
\definecolor{nngrn}{rgb}{0,0.5,0.5}
\definecolor{CitingColor}{rgb}{0,0.3,1}

\newcommand{\blu}{\color{nblue}}

\begin{document}
\title{Quantum complementarity: A novel resource for unambiguous exclusion and encryption}

\author{Chung-Yun Hsieh}
\email{chung-yun.hsieh@bristol.ac.uk}
\affiliation{H. H. Wills Physics Laboratory, University of Bristol, Tyndall Avenue, Bristol, BS8 1TL, UK}

\author{Roope Uola}
\affiliation{Department of Applied Physics, University of Geneva, 1211 Geneva, Switzerland}

\author{Paul Skrzypczyk}
\affiliation{H. H. Wills Physics Laboratory, University of Bristol, Tyndall Avenue, Bristol, BS8 1TL, UK}
\affiliation{CIFAR Azrieli Global Scholars program, CIFAR, Toronto, Canada}

\date{\today}

\begin{abstract}
{\em Complementarity} is a phenomenon explaining several core features of quantum theory, such as the well-known uncertainty principle.
Roughly speaking, two objects are said to be complementary if being certain about one of them necessarily forbids useful knowledge about the other.
Two quantum measurements that do not commute form an example of complementary measurements, and this phenomenon can also be defined for ensembles of states.
Although a key quantum feature, it is unclear whether complementarity can be understood more operationally, as a {\em necessary} resource in some quantum information task. Here we show this is the case, and relates to a novel task which we term {\em $\eta$-unambiguous exclusion}. 
As well as giving complementarity a clear operational definition, this also uncovers the foundational underpinning of unambiguous exclusion tasks for the first time. We further show that a special type of measurement complementarity is equivalent to advantages in certain {\em encryption tasks}. Finally, our analysis suggest that complementarity of measurement and state ensemble can be interpreted as strong forms of measurement incompatibility and quantum steering, respectively.
\end{abstract}

\maketitle

Complementarity is  one of the most important features of quantum theory.
Conceptually, it states that being certain about one property completely blocks us from having or obtaining information about another property.
The most well-known facet of this phenomenon is the uncertainty principle~\cite{Busch2014RMP}.
This is one of the fundamental differences between quantum and classical measurements, with quantum observables not commuting in general, such that two properties can, quantum mechanically, be mutually exclusive, i.e., {\em complementary} to each other~\cite{Incom-review}. 
For instance, in a single electron system, spin measurement of two orthogonal components, such as $S_x$ and $S_z$, form a pair of complementary measurements. 
If we are certain about the outcome of $S_x$, then the outcome of $S_z$ is necessarily completely random. 

Complementarity can also be viewed as a phenomenon of {\em sets of quantum states}.
For instance, two bases of pure states can be mutually exclusive, i.e., no common information is shared between these bases~\cite{Durt2010}.
While being a foundational underpinning of quantum theory, it is an open question whether quantum complementarity can be equivalently defined and understood more operationally, through an appropriate quantum information task, or tasks.
A clear answer to this question would not only bridge a fundamental quantum feature with quantum information science, but also identify quantum complementarity as a novel type of quantum resource,  {\em necessary} in certain operational applications.
This work aims to fully explore and address the operational relevance of quantum complementarity, and so show that it in fact has an operational significance in the context of a novel task ---  which we term {\em unambiguous exclusion}.

\begin{figure*}
	\scalebox{1.0}{\includegraphics{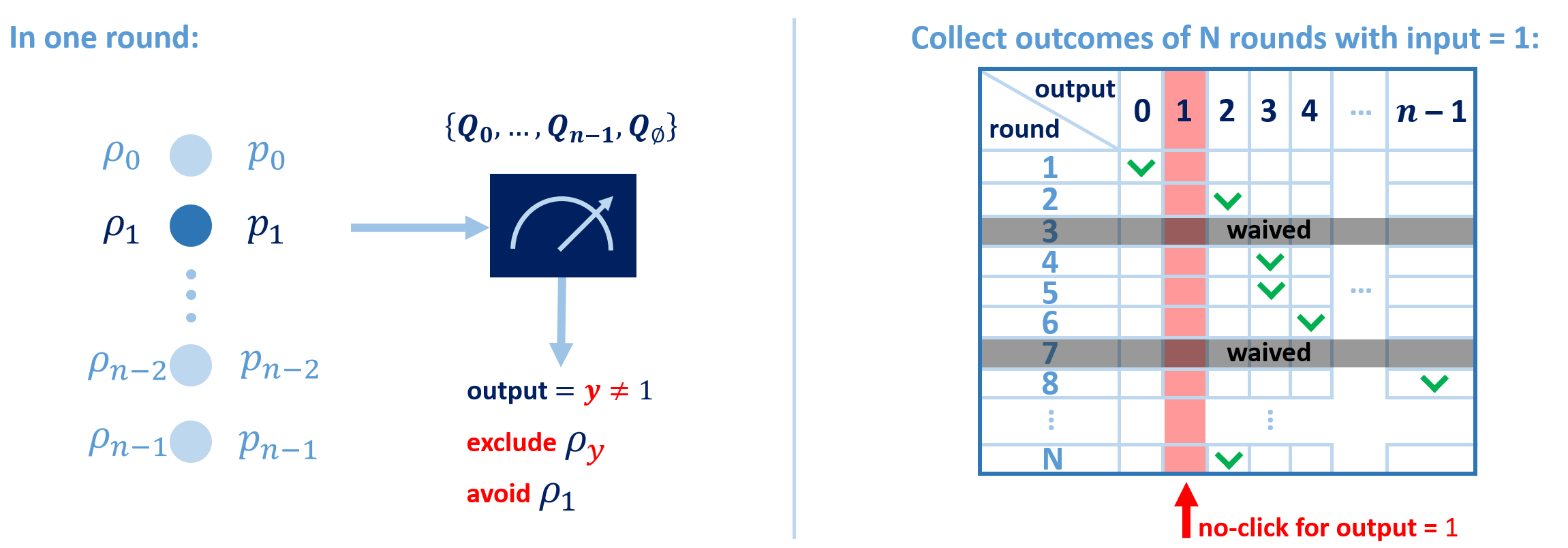}}
	\caption{
{\bf A schematic representation of $\eta$-unambiguous state exclusion.} 
One of $\cardin$ states $\rho_x$ is selected at random, with probability $p_x$.
The goal is to exclude some state $\rho_y$ that is {\em not} the selected state $\rho_x$.
The player applies a \mbox{$(\cardin+1)$-outcome} measurement $\{Q_0, Q_1, \ldots, Q_{\cardin-1}, Q_\emptyset\}$, with the aim of minimising the error probability $\sum_x p_x {\rm tr}(Q_x\rho_x)$, i.e., the probability of \emph{incorrectly excluding} the state $\rho_x$.
More precisely, for each selected $\rho_x$, $\sum_{y\neq x,\emptyset}p_y{\rm tr}(Q_y\rho_y)$ is the probability of having some measurement outcomes guaranteed {\em not} to be $x$ --- the probability of successfully {\em avoiding} $\rho_x$.
As shown in the figure, $\rho_1$ is avoided via some measurement outcome not equal to $1$.
The player is allowed to waive a round (provide no answer), which corresponds to the measurement outcome of $Q_\emptyset$. We require the task to be $\eta$-unambiguous --- meaning that rounds can only be waived with probability at most $\eta$. This translates into the condition $Q_\emptyset\le\eta\id$.
To better visualise the idea of exclusion, let us collect outcomes of $N$ rounds with a {\em fixed} input, say, $x=1$.
In this case, outcome data can be listed with a blank column below \mbox{output $=1$;} namely, $x=1$ can never occur, and it is never incorrectly excluded.
}
	\label{Fig:USE} 
	\end{figure*}

\section*{Quantum complementarity of measurements and state ensembles}
In quantum theory, the most general (destructive) measurement is described by a {\em positive operator-valued measure} (POVM)~\cite{QIC-book}, which is a set $\{E_a\}_a$ of operators satisfying $E_a\ge0\;\forall\;a$ and $\sum_aE_a=\id$.
When we consider a state $\rho$, this POVM describes a measurement scheme whose measurement outcome is $a$ with probability $P(a|\rho) := {\rm tr}(\rho E_a)$. 
In this work, we call a set of POVMs ${\bf E}\coloneqq\{E_{a|x}\}_{a,x}$ a {\em measurement assemblage}, where $\{E_{a|x}\}_a$ is a POVM for each $x$.
Conceptually, a measurement assemblage ${\bf E}$ is said to be {\em complementary} if it is impossible to say anything non-trivial  about every measurement {\em simultaneously}. 
This can be formalised by saying that if we pick even just a single outcome $a_x$ of each measurement, it is impossible to find a non-trivial \textit{state-dependent} lower bound $0 < \ell(\rho) \leq P(a_x|x,\rho) = \tr(\rho E_{a_x|x})$ that would hold for all measurements indexed by $x$. 
By non-trivial, we mean that $\ell(\rho)$ does not vanish identically for all $\rho$, in which case this would just tell us the trivial fact that probabilities are non-negative. When this is not the case, there is at least one state for which it is possible to simultaneously obtain non-trivial information about the entire measurement assemblage. Assuming that $\ell(\rho)$ is linear, it can be written as $\ell(\rho) = \tr(M \rho)$, with $M \geq 0$. A mathematically equivalent formulation of measurement complementarity is that ${\bf E}$ is {\em complementary} if, for every possible set of outcomes $\{a_x\}_x$,
\begin{align}\label{Eq:Def}
E_{a_x|x}\ge M\quad\forall\;x\quad\Rightarrow\quad M=0.
\end{align}

We can also discuss complementarity of state ensembles.
In a bipartite setup, when two parties share a state $\rho_{AB}$, a measurement assemblage ${\bf E}$ implemented by $A$ will induce a set of (sub-normalised) states in $B$:
$
\sigma_{a|x}({\bf E},\rho_{AB})\coloneqq{\rm tr}_A\left[(E_{a|x}\otimes\id_B)\rho_{AB}\right].
$
The set ${\bm\sigma}({\bf E},\rho_{AB})\coloneqq\{\sigma_{a|x}({\bf E},\rho_{AB})\}_{a,x}$ (or simply ${\bm\sigma} = \{\sigma_{a|x}\}_{a,x}$) is called a {\em state assemblage} induced by ${\bf E}$ and $\rho_{AB}$.
It can be viewed as a collection of {\em state ensembles}: for each $x$, we have a state ensemble 
$
{\bm\sigma}_x\coloneqq\{\sigma_{a|x}\}_a,
$
specified by a collection of {\em sub-normalised} states, such that the state $\rho_{a|x}~:=~\sigma_{a|x}/{\rm tr}(\sigma_{a|x})$ occurs with the probability $P(a|x)~=~{\rm tr}(\sigma_{a|x})$ after the local measurement of $A$.
Notably, local measurements cannot change $B$'s local state, due to non-signalling: we have $\sum_a\sigma_{a|x} = {\rm tr}_A(\rho_{AB})$ for every $x$.

Similar to the case of measurements, we say that a state assemblage 
${\bm\sigma}$ is complementary if it is impossible to say anything non-trivial simultaneously about every ensemble it encompasses. 
In line with above, this is formalised by saying that even if we pick a single state from each ensemble, specified by ${a_x}$, then it is impossible to find a non-trivial {\em measurement-dependent} lower bound $0 < L(N_0) \leq P(a_x,0|x,N_0) = \tr(\sigma_{a_x|x}N_0)$, where $N_0$ is a POVM element, and we have (arbitrarily) labelled the corresponding outcome $0$. 
Assuming again that $L(N_0)$ is linear, it can be written as $L(N_0) = \tr(N_0 \omega)$ with $\omega \geq 0$, and a mathematically equivalent formulation of state ensemble complementarity is that ${\bm\sigma}$ is complementary if, for every set
$\{a_x\}_x$, 
\begin{align}\label{Eq:Def-State}
\sigma_{a_x|x}\ge \omega\quad\forall\;x\quad\Rightarrow\quad \omega=0.
\end{align}
In what follows, we will show that we have a unified, operational way of analysing these two notions of complementarity for measurements and state ensembles. 
To do so, we will first introduce a novel task, termed {\em $\eta$-unambiguous exclusion}, before showing how it can be applied to complementarity.

\section*{Exclusive mutual information and its operational meaning}
The formulations in Eqs.~\eqref{Eq:Def} and~\eqref{Eq:Def-State} provide a binary characterisation of when a set of measurements or ensembles is complementary or not. 
They however also provide a natural way to \emph{quantify} or \emph{measure} how close to being complementary a measurement or state assemblage is. 
In particular, we can ask what is the largest trace of $M$ and $\omega$ that is possible in Eqs.~\eqref{Eq:Def} and~\eqref{Eq:Def-State}, respectively. 
This motivates us to introduce the following function, defined for a set of positive semi-definite operators $\{N_x\}_x$:
\begin{align}\label{Eq:C-function}
q_{\rm exc}(\{N_x\}_x)\coloneqq \max_{\substack{N_x\ge P\,\,\forall\,x\\P\ge0}}{\rm tr}(P),
\end{align}
which is a semi-definite program~\cite{Watrous-book,SDP-textbook} (see Appendix for its dual).
Note that 
$
0\le q_{\rm exc}(\{N_x\}_x) \le\min_x{\rm tr}(N_x), 
$
and the upper bound is saturated if $N_x = N_y$ $\forall\;x,y$.
Consequently, when we consider $\{N_x\}_x$ to be a set of quantum states, $I_{\rm exc}:= -\log_2 q_{\rm exc}$ behaves as a type of \emph{mutual information} (see Appendix). 
For this reason, we call $I_{\rm exc}$ the {\em exclusion mutual information}.
Our first new insight is that in this context, $I_{\rm exc}$ has a clear operational meaning in a novel task, aiming to (almost) {\em unambiguously exclude} one of $\rho_x$.
We call this an {\em $\eta$-unambiguous state-exclusion task}, which we sketch here (see Fig.~\ref{Fig:USE} for the detailed setup).

During the task, $\rho_x$ is prepared with probability $p_x>0$, for $x = 0, \ldots, \cardin-1$.
The aim (as in all exclusion tasks, which are also known as {\em antidistinguishability} and {\em state elimination tasks}; see, e.g., Refs~\cite{Mishra2023,Bandyopadhyay2014,PhysRevResearch.2.013326,Russo2023,Pusey2012,Barrett2014,Ducuara2020} and the references therein) is to make a guess of which state was  {\em not} prepared. In line with `unambiguous' tasks, we allow the player to waive rounds, and declare `I don't know' some of the time. 
In \emph{strict} unambiguous tasks, no error is tolerated, i.e., the player must only answer when they are sure to be correct, and the goal then is to minimise the fraction of inconclusive rounds. Here, we introduce a novel non-strict version of unambiguous tasks; in particular, we bound the fraction of inconclusive rounds by $\eta$ (which we refer to as the \emph{inconclusiveness}). 
This in general will mean that there will be errors, and so we minimise the chance of making an error. 
These $\eta$-unambiguous tasks are therefore similar to hypothesis tests~\cite{Wang2012}, in having two types of `error': error in guessing, and `\emph{inconclusion}', which trade-off against each other.\footnote{That is, if no error is permitted (as in strict unambiguous exclusion), then the inconclusion will be largest; on the other hand, bounding the inconclusion means that there will in general be some amount of error in exclusion.}

Putting everything together and optimising over all possible strategies (i.e., over the measurement used in order to carry out the task), with bounded inconclusion, the minimal error probability reads
\begin{align}\label{Eq:P-error}
\Perrorstate \left(\stEns,\prob,\eta\right) \coloneqq \min_{\substack{(1-\eta)\id\le\sum_xQ_x\le\id\\Q_x\ge0\;\;\forall\;x}}\sum_{x=0}^{\cardin-1} p_x{\rm tr}(Q_x\rho_x),
\end{align}
where $\stEns\coloneqq\{\rho_x\}_x$ and $\prob\coloneqq\{p_x\}_x$.
In order to properly quantify how difficult the task is, it is necessary to consider also the minimal error probability in the purely \emph{classical} version of this task, which is \emph{unambiguous exclusion of (classical) information}. 
In this setting, the player knows that the random variable $X$ is distributed according to $p_x$, and they must exclude some value $y\neq x$ to avoid $x$, with inconclusion bounded by $\eta$. The optimal strategy in this task is rather straightforward: the player simply waives each round with probability $\eta$, and with probability $1-\eta$ they always exclude the least likely value, i.e., $x =~{\rm argmin}\{p_y\}$. The minimal (classical) error probability is 
\begin{align}
\Perrorclass\left(\prob,\eta\right)\coloneqq~(1~-~\eta)~\min_xp_x.
\end{align} 

We can now state the operational interpretation of $q_{\rm exc}$ (and hence also $I_{\rm exc}$). 
From now on, $\prob>0$ denotes that $p_x>0$ $\forall\,x$.

\begin{result}\label{Result:ERE-Task-maintext}
For every set of states $\stEns = \{\rho_x\}_{x=0}^{\cardin-1}$, there exists a critical inconclusiveness $\eta_*<~1$ such that 
\begin{align}\label{Eq:Result:ERE-Task}
q_{\rm exc}\left(\stEns\right) = \min_{\prob>0}\frac{\Perrorstate\left(\stEns,\prob,\eta\right)}{\Perrorclass\left(\prob,\eta\right)}\quad\forall\;\eta_*\le\eta<1.
\end{align}
The minimum is achieved by $\left\{p_x = 1/\cardin\right\}_x$.
\end{result}
We detail the proof in the Appendix.
This shows that $I_{\rm exc} = -\log_2 q_{\rm exc}$ operationally characterises the quantum advantage (i.e., multiplicative decrease compared to the classical error) that is provided by the set of states $\{\rho_x\}_x$ (as a quantum encoding of the classical information $x$) in the task of $\eta$-unambiguous exclusion, in the regime where the inconclusiveness is sufficiently small.  

A natural question is whether the class of strategies considered here for the quantum player is general enough, or whether one can improve the advantage by considering a more general strategy (i.e., one involving the use of additional randomness, and more general guessing strategies after the measurement result). We prove in the Appendix, rather surprisingly, pre-measurement quantum channels and post-processing of the output classical data {\em cannot} improve the advantage.

\section*{Characterising quantum complementarity}

\begin{figure*}
	\includegraphics[width=\textwidth]{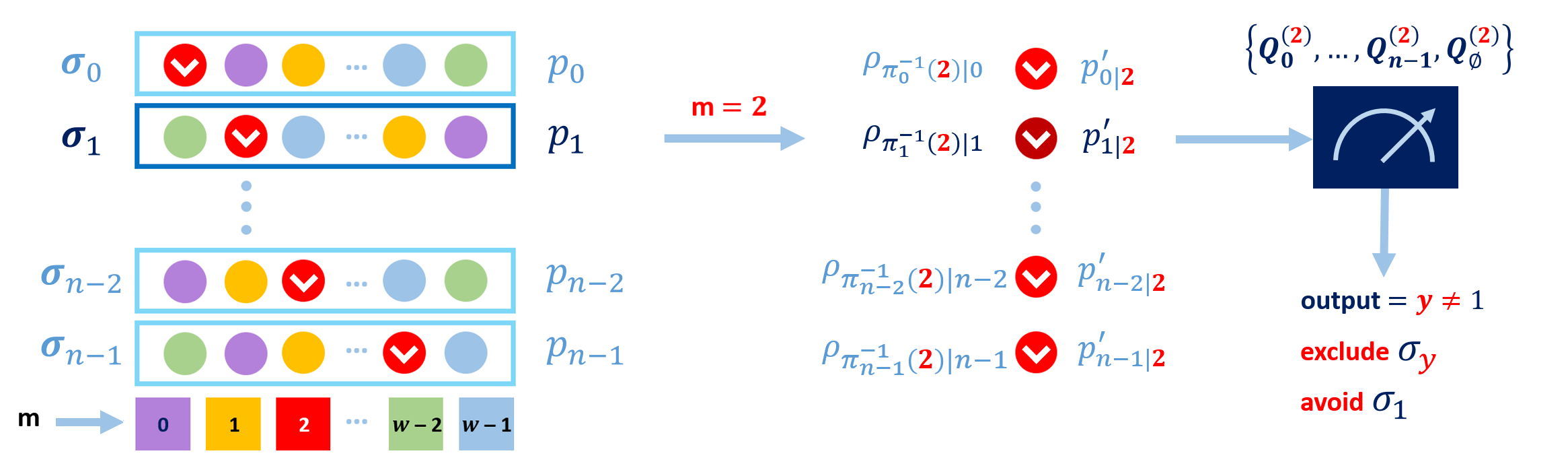}
	\caption{
	{\bf A schematic representation of $\eta$-unambiguous ensemble exclusion.} 
 The classical information $x$ is now encoded in an ensemble ${\bm \sigma_x}$. After the ensemble $\bm\sigma_x$ is chosen, a state $\rho_{a|x}$ from that ensemble is selected, with probability $p(a|x)$. The player is sent this state, along with the additional side information $m$ (depicted as the colour of the state), which encodes conditional information of the form `if $X = x$ then $a = a_x$', via $m = \pi_x(a)$, where $\permu = \{\pi_x\}_x$ is a set of permutations specifying the particular colouring.
 For instance, in the figure, $m=2$ is specified by red with a white click. 
 The player thus effectively plays a state exclusion task, for the ensemble of states $\{\rho_{\pi_x^{-1}(m)|x}\}_x$ with probabilities $p'_{x|m}$ as defined below Eq.~\eqref{Eq:p-ensemble} (namely, it is a state exclusion for states marked in red with white clicks).
 The player can now choose their measurement based upon the message $m$, which is given by $\{Q^{(m)}_0,...,Q^{(m)}_{\cardin-1},Q^{(m)}_\emptyset\}$.
 The goal is to avoid the sent ensemble (i.e., exclude an ensemble different from the sent one), and as before, the player is allowed to waive a given round with probability at most $\eta$. 
	\label{Fig:UensembleSE}}
	\end{figure*}

We now return our attention to complementarity, and show that by considering a different quantum encoding of classical information --- this time into \emph{ensembles} rather than states, we can obtain a result similar to Theorem~\ref{Result:ERE-Task-maintext}. 

Recall from Eq.~\eqref{Eq:C-function} that $q_{\rm exc}$ is designed to quantify the `size' of the operator $P$ in Eq.~\eqref{Eq:Def}.
We can return to this motivation and use it to introduce the following quantifier of measurement complementarity:
\begin{align}
C_{\rm POVM}({\bf E})\coloneqq\max_{\permu}\sum_{a=0}^{\cardout-1}q_{\rm exc}(\{E_{\pi_x(a)|x}\}_x),
\end{align}
where, throughout this paper, $\permu\coloneqq\{\pi_x\}_x$ denotes all possible sets of permutations; i.e., each $\pi_x$ permutes the label $a=0,...,\cardout-1$, 
where $\cardout$ is the cardinality of the set of outcomes $a$.
This definition can be seen as a quantification of Eq.~\eqref{Eq:Def}. 
Indeed, we see immediately that  ${\bf E}$ is complementary if and only if
$
C_{\rm POVM}({\bf E})=0.
$
Moreover, if ${\bf E}$ is close to being complementary, then $C_{\rm POVM}({\bf E})$ will be small.

Exactly the same mathematics also applies for state ensemble complementarity --- for any state assemblage ${\bm\sigma}$, we define

\begin{align}
C_{\rm SE}({\bm\sigma})\coloneqq \max_{\permu}\sum_{a=0}^{\cardout - 1}q_{\rm exc}(\{\sigma_{\pi_x(a)|x}\}_x).
\end{align}
Again, ${\bm\sigma}$ is complementary if and only if $C_{\rm SE}({\bm\sigma})=0$, and if the state assemblage is close to being complementary, then $C_{\rm SE}$ will be small. 

We can now directly build upon the previous section, to show that $C_{\rm POVM}$ and $C_{\rm SE}$ have clear operational meanings, by generalising $\eta$-unambiguous exclusion from state exclusion to \emph{state-ensemble} exclusion. 
We sketch the task here, and a detailed summary can be found in  Fig.~\ref{Fig:UensembleSE}. The basic idea is to encode the information $x$ into an ensemble, rather than into an individual state. 

Recall that a state assemblage ${\bm\sigma} = \{{\bm\sigma_x}\}_x$ is a collection of state ensembles, labelled by $x$. 
We can therefore naturally view this as an \emph{ensemble-encoding} of the classical information $x$. Consider therefore a task where after encoding $x$ into an ensemble, a player is given a state drawn at random from that ensemble; i.e., they are given the state $\rho_{a|x}=\sigma_{a|x}/{\rm tr}(\sigma_{a|x})$ with probability $P(a|x)={\rm tr}(\sigma_{a|x})$.
In the absence of any other information, this task reduces to state-exclusion of the \emph{average} states $\averho \coloneqq \sum_a P(a|x) \rho_{a|x}$ of the ensembles, and therefore does not add anything new.

In order to arrive at a novel task, we provide the player with additional \emph{conditional} information of the form: `if $X = x$, then $a = a_x$'.
The ensemble {\em seen by the player} is thus $\{\rho_{a_x|x}\}_x$.
The additional information received by the player can be considered as a type of classical side-information, along with the quantum side-information $\rho_{a_x|x}$, that the player can use to potentially do better in the task of $\eta$-unambiguous exclusion, compared to state exclusion. 

In order to analyse the error probability, it is useful to express the information received in a mathematically-equivalent formulation, which is easier to analyse. 
We can instead model the situation as there being an agreement before the task starts of a set of permutations $\permu=\{\pi_x\}_x$, with the player receiving the classical information $m = \pi_x(a)$. 
The player then knows that, if $X = x$, then the value of $a$ was $a_x = \pi_x^{-1}(m)$, i.e., they can recover the conditional information from above, by applying the inverse permutation to the message they received. 

To play the game, the player now proceeds as previously, by measuring the state received, using a measurement that will announce an inconclusive result with probability at most $\eta$. 
The main difference now is that the player's measurement can depend upon the message $m$, and the final error probability will be averaged over the different possible messages. 
Altogether, we see that
\begin{multline}\label{Eq:p-ensemble}
\perrorens({\bm\sigma},\prob,\eta, \permu)\coloneqq\\
\sum_{m=0}^{\cardout-1} p'_m \Perrorstate\left(\{\rho_{\pi_x^{-1}(m)|x}\}_x,\{p'_{x|m}\}_x,\eta\right),
\end{multline}
where $p'_m = \sum_x p_x P\left(\pi_x^{-1}(m)\middle|x\right)$ is the probability that the player receives the message $m$, and
\mbox{$
p'_{x|m} = p_x P\left(\pi_x^{-1}(m)\middle|x\right)/p'_m
$}
is the (conditional) probability of the states in the ensemble $\{\rho_{\pi_x^{-1}(m)|x}\}_x$.

Finally, we will be interested in the \emph{least-informative} (or \emph{worst-case}) classical side information, that is, the set of permutations $\permu$ such that the error probability is largest:
\begin{align}\label{Eq:Perrorens definition}
	\Perrorens({\bm\sigma},\prob,\eta) \coloneqq \max_{\permu} \perrorens({\bm\sigma},\prob,\eta,\permu).
\end{align}
With the above in place, we can now generalise Theorem~\ref{Result:ERE-Task-maintext} to the case of $\eta$-unambiguous ensemble exclusion with \emph{non-signalling} assemblages:
\begin{result}\label{Result:Exclusion-SE-maintext}
For every non-signalling state assemblage ${\bm\sigma}$ containing $\cardin$ ensembles, there exists $\eta_*<1$ such that 
\begin{align}\label{Eq:Result:Exclusion-SE}
C_{\rm SE}({\bm\sigma}) = \min_{\prob>0}\frac{\Perrorens\left({\bm\sigma},\prob,\eta\right)}{\Perrorstate \left(\{\averho\}_x,\prob,\eta\right)}\quad\forall\,\eta_*\le\eta<1.
\end{align}
The minimum is achieved by  $\left\{p_x = 1/\cardin\right\}_x$.
\end{result}
See Appendix for the proof.
Theorem~\ref{Result:Exclusion-SE-maintext} says that if we consider the worst-case advantage offered by a non-signalling assemblage $\bm\sigma$ in the setting of ensemble-exclusion over induced state-exclusion task (based upon the average state of each ensemble) is precisely given by the \emph{complementarity of the state assemblage} $\bm\sigma$. 
Said the other way, this shows that this quantification of complementarity has an operational interpretation, as characterising the advantage a state assemblage provides in the newly introduced task of $\eta$-unambiguous exclusion. 

We note first that, since we focus on non-signalling assemblages, $\averho = \rho_B$ for every $x$, and contains \emph{no information} about $x$. We therefore see that for such state assemblages $\Perrorstate \left(\{\averho\}_x,\prob,\eta\right) = \Perrorclass(\prob,\eta)$. 
The advantage is also thus the advantage over classical $\eta$-unambiguous exclusion. 
Nevertheless, the form presented in Eq.~\eqref{Eq:Result:Exclusion-SE} appears to be the most physically relevant formulation of the result. 

Second, we recall that a state assemblage $\bm\sigma$ is complementary when $C_{\rm SE}({\bm\sigma}) = 0$, corresponding to the impossibility of saying anything non-trivial simultaneously about every ensemble it contains. Theorem \ref{Result:Exclusion-SE-maintext} shows that such state ensembles have \emph{zero error} in $\eta$-unambiguous ensemble exclusion; namely, they are the strongest resources. 

Finally, in Appendix we show a no-go result: adding pre-measurement quantum channels and post-processing of the output classical data {\em cannot} improve the performance in ensemble exclusion.

It turns out that we can also relate $C_{\rm POVM}$ and $C_{\rm SE}$, consequently linking measurement complementarity and $\eta$-unambiguous exclusion tasks.
In what follows, \mbox{$\rho_A\coloneqq{\rm tr}_B(\rho_{AB})$,} and $\mu_{\rm min}(\rho_A)$ is the smallest eigenvalue of the state $\rho_A$.
\begin{result}\label{Result:ExclusionTask}
Let ${\bf E}$ be a measurement assemblage.
Then for every $\rho_{AB}$ with full-rank marginal states, we have
\begin{align}
\mu_{\rm min}(\rho_A)C_{\rm POVM}({\bf E})\le C_{\rm SE}[{\bm\sigma}({\bf E},\rho_{AB})].
\end{align}
Moreover, there always exist a state $\rho_{AB}$ with full-rank marginals achieving the upper bound.
\end{result}
See Appendix for the proof.
Note that we consider $\rho_{AB}$ with full-rank marginals without loss of generality.
Physically, if $\rho_{AB}$ has non-full-rank marginals, one can effectively treat it as a state in a smaller Hilbert space.
Theorem~\ref{Result:ExclusionTask} implies the following observation:
\begin{maincoro}\label{coro:Complementary1-1mapping-1}
When ${\bf E}$ is not complementary, then ${\bm\sigma}({\bf E},\rho_{AB})$ is not complementary for every $\rho_{AB}$ with full-rank marginals.
Conversely, when ${\bf E}$ is complementary, then there always exists at least one $\rho_{AB}$ such that ${\bm\sigma}({\bf E},\rho_{AB})$ is complementary.
\end{maincoro}
Our findings thus provides a {\em one-to-one correspondence} between measurement complementarity and state ensemble complementarity.
Moreover, this link has a clear operational meaning in unambiguous exclusion tasks of state ensembles.

\section*{Application to encryption tasks}
As an application of the above results, measurement complementarity of {\em a pair of} measurements is in fact equivalent to advantages in an {\em encryption task}, which we detail now.
Consider two agents $A$ (sender) and $B$ (receiver) sharing a maximally entangled state $\ket{\phi^+_{AB}}\coloneqq\sum_{n=0}^{d-1}\ket{n}_{A}\otimes\ket{n}_B/\sqrt{d}$.
Their goal is to transmit one bit of classical information (described by the index set $\{0,1\}$)  from $A$ to $B$ 
through a {\em public} communication channel such that (i) ({\em encryption}) the transmitted information cannot be known by the public, and (ii) ({\em unambiguous}) $B$ wants to always decode correctly.
$B$ will be allowed to waive the round, but when $B$ decodes, it must be correct.
To encode and hide the classical information, $A$ chooses a pair of measurements ${\bf E} = \{E_{a|x}\}_{a,x}$ with $x=0,1$ ({\em encoding strategy}) and a pair of permutation $\permu=\{\pi_0,\pi_1\}$ ({\em encryption strategy}).
To extract the information, $B$ chooses between $\cardout$ three-outcome measurements ${\bf Q}_{m}^\eta = \{Q_{0|m},Q_{1|m},Q_{\emptyset|m}\}$ with $m=0,...,\cardout-1$ ({\em decoding strategy}) based on the knowledge of ${\bf E}$ and $\permu$. 
We demand that $Q_{\emptyset|m}\le \eta\id$ $\forall\;m=0,...,\cardout-1$ to make sure the probability of being inconclusive is bounded.
In each round, $x=0,1$ is announced uniformly ($p_x = 1/2$ for $x=0,1$).
$A$ measures their half of $\ket{\phi^+_{AB}}$, using the measurements ${\bf E}$, and post-processes the measurement outcome $a$ by $\pi_x$, obtaining $m\coloneqq\pi_x(a)$.
$A$ then sends the {\em message} $m$ to $B$ through a public classical communication.
$B$ then measures their half of $\ket{\phi^+_{AB}}$ using ${\bf Q}_{m}^\eta$ and decodes $x$ as $0$ or $1$, if they measure $Q_{0|m}$ or $Q_{1|m}$, respectively.
$B$ waives the round if they obtain the outcome corresponding to $Q_{\emptyset|m}$.
We call this an {\em entanglement-assisted $\eta$-encryption task} of ${\bf E}$.
See Fig.~\ref{Fig} for a schematic understanding.

\begin{figure}
\scalebox{1.0}{\includegraphics{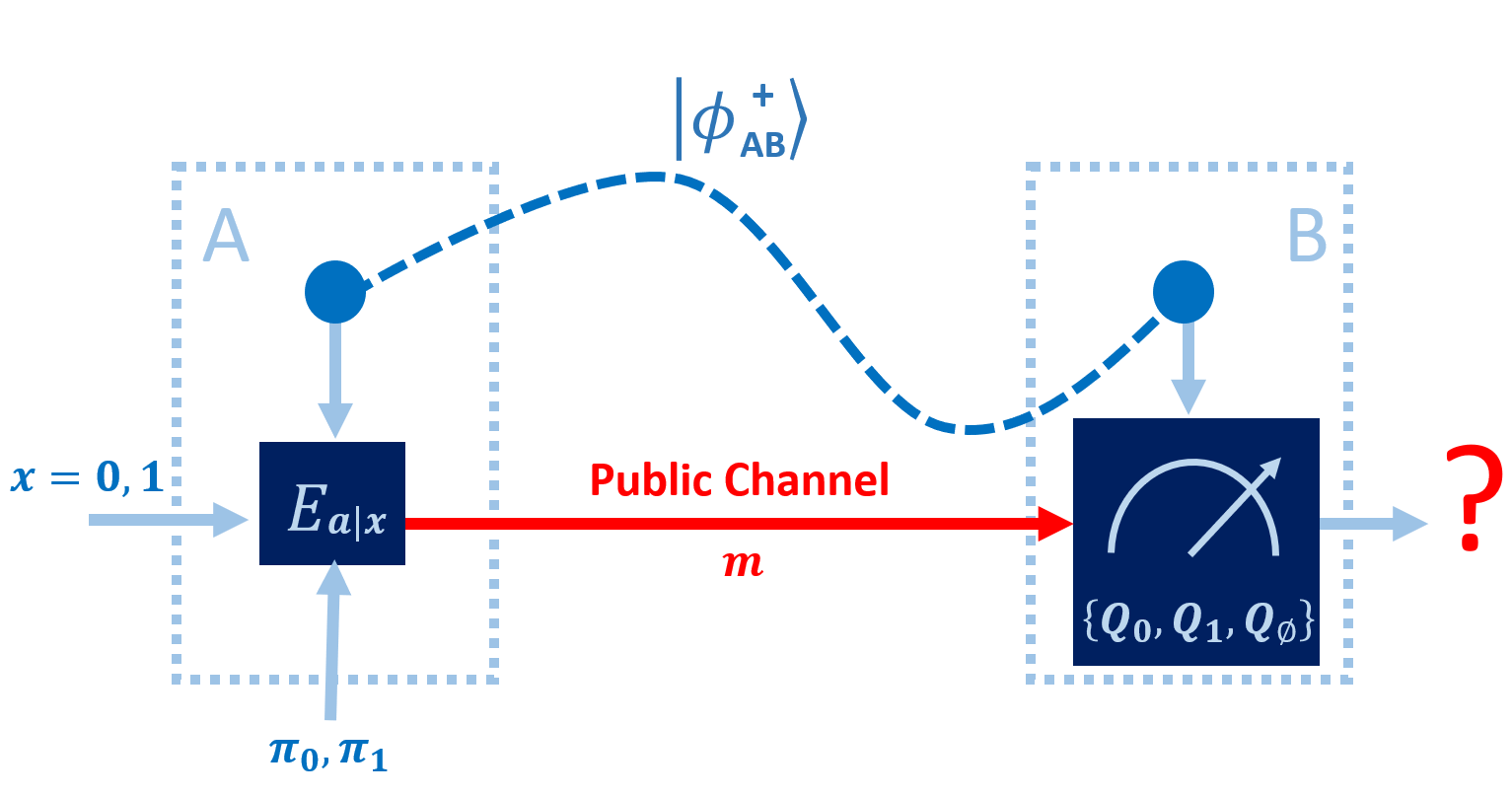}}
\caption{
{\bf A schematic representation of the  entanglement-assisted $\eta$-encryption task.}
At the beginning of each round, $x=0,1$ is announced uniformly.
When $x$ is the input, $A$ measures half of $\ket{\phi^+_{AB}}$ through with $\{E_{a|x}\}_a$, and post-processes the measurement outcome $a$ by the encryption strategy $\pi_x$, mapping it into $m\coloneqq\pi_x(a)$.
$A$ then sends the {\em message} $m$ to $B$ through a public communication channel.
After receiving the message, $B$ measures their half of $\ket{\phi^+_{AB}}$ using ${\bf Q}_{m}^\eta$ and guesses $x$ is $0,1$ if they obtain measurement outcomes of  $Q_{0|m},Q_{1|m}$, respectively.
Combining everything, one can see why Eq.~\eqref{Eq:p-error} provides the error probability in this task. 
Finally, $B$ waives the round if the measurement outcome is from $Q_{\emptyset|m}$, whose probability is bounded by the inconclusiveness $\eta$.
}
\label{Fig} 
\end{figure}

Now we analyse the decoding statistics for a given pair of measurements ${\bf E}$.
Similar to the case of Eq.~\eqref{Eq:p-ensemble}, with a given encryption strategy $\permu$, $B$ receives the message $m$ with probability $p'_m = \sum_x p_x P\left(\pi_x^{-1}(m)\middle|x\right)=\sum_{x=0,1}{\rm tr}\left(E_{\pi^{-1}_{x}(m)|x}\right)/2d$. 
In this case, $B$ sees an effective state ensemble consisting of two states (i.e., for $x=0,1$) $\rho_{\pi_x^{-1}(m)|x} = {\rm tr}_A\left[\left(E_{\pi^{-1}_{x}(m)|x}\otimes \id_B\right)\proj{\phi_{AB}^+}\right]/P\left(\pi_x^{-1}(m)\middle|x\right)$, each with probability $p'_{x|m} = p_x P\left(\pi_x^{-1}(m)\middle|x\right)/p'_m$.
By choosing the best decoding strategy for each message $m$, $B$ has the following optimal achievable error probability with a given encryption strategy $\permu$:
\begin{multline}\label{Eq:p-error}
\perrorencrypt\left({\bf E},\permu,\eta\right)\coloneqq\\
\sum_{m=0}^{\cardout-1}p'_m\min_{\substack{(1-\eta)\id\le Q_0+Q_1\\Q_0+Q_1\le\id\\ Q_0\ge0, Q_1\ge0}}\sum_{x=0,1}p'_{x|m}{\rm tr}\left(\rho_{\pi_x^{-1}(m)|x}Q_x\right)\\
=\sum_{m=0}^{\cardout-1}\min_{\substack{(1-\eta)\id\le Q_0+Q_1\\Q_0+Q_1\le\id\\ Q_0\ge0, Q_1\ge0}}\frac{1}{2}\sum_{x=0,1}\bra{\phi_{AB}^+}\left(E_{\pi^{-1}_{x}(m)|x}\otimes Q_x\right)\ket{\phi_{AB}^+}.
\end{multline}
Maximising over all possible encryption strategies, we obtain 
\begin{align}
\Perrorencrypt({\bf E},\eta)
\coloneqq\max_{\permu}\perrorencrypt\left({\bf E},\permu,\eta\right),
\end{align}
which is the {\em worst-case (biggest) optimal average error probability} among all possible encryption strategies of $A$.
One can check that $\Perrorencrypt({\bf E},\eta)=0$ {\em if and only if we have unambiguous $\eta$-encryption for every encryption strategy}.
Namely, {\em for every} sent message $m$ {\em and} encryption strategy $\permu$, $B$ can always choose a decoding strategy ${\bf Q}_{m}^\eta$ to {\em faithfully} decode the classical information $x$, whenever they choose to decode.

It turns out that this encryption task is closely related to measurement complementarity, as we detail below:
\begin{result}\label{Result:Encryption}
Let ${\bf E}$ be a pair of measurements. Consider a bipartite setting with equal local dimension $d$.
There exists $\eta_*<1$ such that, for every $\eta_*\le\eta<1$,
\begin{align}
	\Perrorencrypt({\bf E},\eta) = \frac{1-\eta}{2d}C_{\rm POVM}({\bf E}).
\end{align}

\end{result}
See Appendix for the proof.
Consequently, being {\em unambiguous} for {\em all} encryption strategies in entanglement-assisted $\eta$-encryption tasks is a {\em necessary and sufficient} condition for {\em a pair of} measurements to be complementary.
Theorem~\ref{Result:Encryption} not only provides a task-based operational interpretation of measurement complementarity, but also suggests the following resource equality:
\begin{equation*}
\text{complementarity} + \text{entanglement} = \text{encryption advantage}.
\end{equation*}

A natural question is whether the unambiguous task can also be {\em deterministic}.
This is, however, {\em impossible}: in Appendix we prove that, for every pair of measurements ${\bf E}$, 
\begin{align}\label{Eq:Result:No-Go}
\Perrorencrypt({\bf E},\eta=0) > 0.
\end{align}
Hence, it is impossible to be unambiguous and deterministic for {\em all possible} encryption strategies $\permu$ and sent message $m$.
Note that, however, if $A$ always uses a single {\em fixed} encryption strategy, then it is possible to be unambiguous and deterministic at the same time.

As a simple example to illustrate Theorem~\ref{Result:Encryption}, consider Pauli $Z$ and $X$ measurements, i.e.,
$
E_{0|0} = \proj{0}, E_{1|0}=\proj{1}, E_{0|1} = \proj{+}, E_{1|1}=\proj{-}.
$
They provide unambiguous $\eta$-encryption, $\Perrorencrypt({\bf E},\eta) = 0$, with
$
2/(1+\sqrt{2})\le\eta<1
$
(see Appendix).
Hence, Theorem~\ref{Result:Encryption} shows that this constitutes a proof of complementarity.
When $\eta=2/(1+\sqrt{2})$, the optimal decoding measurements for the encryption strategy $\permu=\{\pi_0,\pi_1\}$ and message $m$  are (see Appendix)
\begin{align}\label{Eq:DecodingStrategy}
Q_{0|m} = \frac{\sqrt{2}\left(\id-E_{\pi^{-1}_1(m)|1}\right)^\intercal}{1+\sqrt{2}},Q_{1|m} = \frac{\sqrt{2}\left(\id-E_{\pi^{-1}_0(m)|0}\right)^\intercal}{1+\sqrt{2}},
\end{align}
where $(\cdot)^\intercal$ is the transpose map.
Finally, for this level of inconclusiveness $\eta$, the probability of successfully decoding is 
$
\sum_{m}P(x_{\rm out} = x,m)\approx0.2929 
$
(see Appendix).

\section*{Relations with measurement incompatibility and quantum steering}
Finally, we show that quantum complementarity can be viewed as a strong notion of {\em measurement incompatibility}~\cite{Incom-review} and {\em quantum steering}~\cite{UolaRMP2020,Cavalcanti_2017-review}.
We start with the case of measurements.
A measurement assemblage ${\bf E}$ is said to be {\em compatible} (or \emph{jointly measurable}) if there exists a `parent' POVM $\{G_\lambda\}_\lambda$ and probability distributions $\{P(a|x,\lambda)\}_{a,x,\lambda}$ such that
$
E_{a|x} = \sum_\lambda P(a|x,\lambda)G_\lambda$ $\forall\,a,x.
$
Namely, a single measurement plus classical post-processing can implement ${\bf E}$.
${\bf E}$ is said to be {\em incompatible} if it is not compatible.
One way to quantify incompatability is via the so-called {\em incompatibility weight}~\cite{PuseyJOPSA15} ${\rm IW}({\bf E})\coloneqq1-\max\{0\le q\le1\;|\;E_{a|x}\ge q L_{a|x}\;\forall\;a,x\;\&\;\{L_{a|x}\}_{a,x} \text{ is compatible}\}$.
Roughly speaking, $1-{\rm IW}({\bf E})$ measures the largest amount of compatibility contained in ${\bf E}$ through a convex mixture.
Hence, ${\rm IW}({\bf E})=0$ if and only if ${\bf E}$ is compatible.
As proved in Appendix, ${\rm IW}$ is in fact related to $C_{\rm POVM}$ as follows:
\begin{result}\label{Result:IncompatibilityWeight}
For every $d$-dimensional measurement assemblage ${\bf E}$, with $\cardin$ inputs and $\cardout$  outcomes, we have
\begin{align}
1-{\rm IW}({\bf E})\le\frac{\cardout^\cardin}{d}C_{\rm POVM}({\bf E}).
\end{align}
\end{result}
This shows that measurement complementarity implies measurement incompatibility. 
On the other hand, it provides a quantitative relationship, further justifying the claim that measurement complementarity is a strong notion of measurement incompatibility.
Finally, using the equivalence between quantum steering and measurement incompatibility~\cite{Uola2015PRL,UolaRMP2020}, Theorem~\ref{Result:ExclusionTask} and Theorem~\ref{Result:IncompatibilityWeight} shows that the notion of ensemble complementarity introduced here can also be viewed as a new, strong notion of quantum steering.
In particular, we see that maximal steerability can be understood as equivalent to the ability of an assemblage to perform perfect unambiguous ensemble exclusion.

\begin{figure}
\scalebox{1.0}{\includegraphics{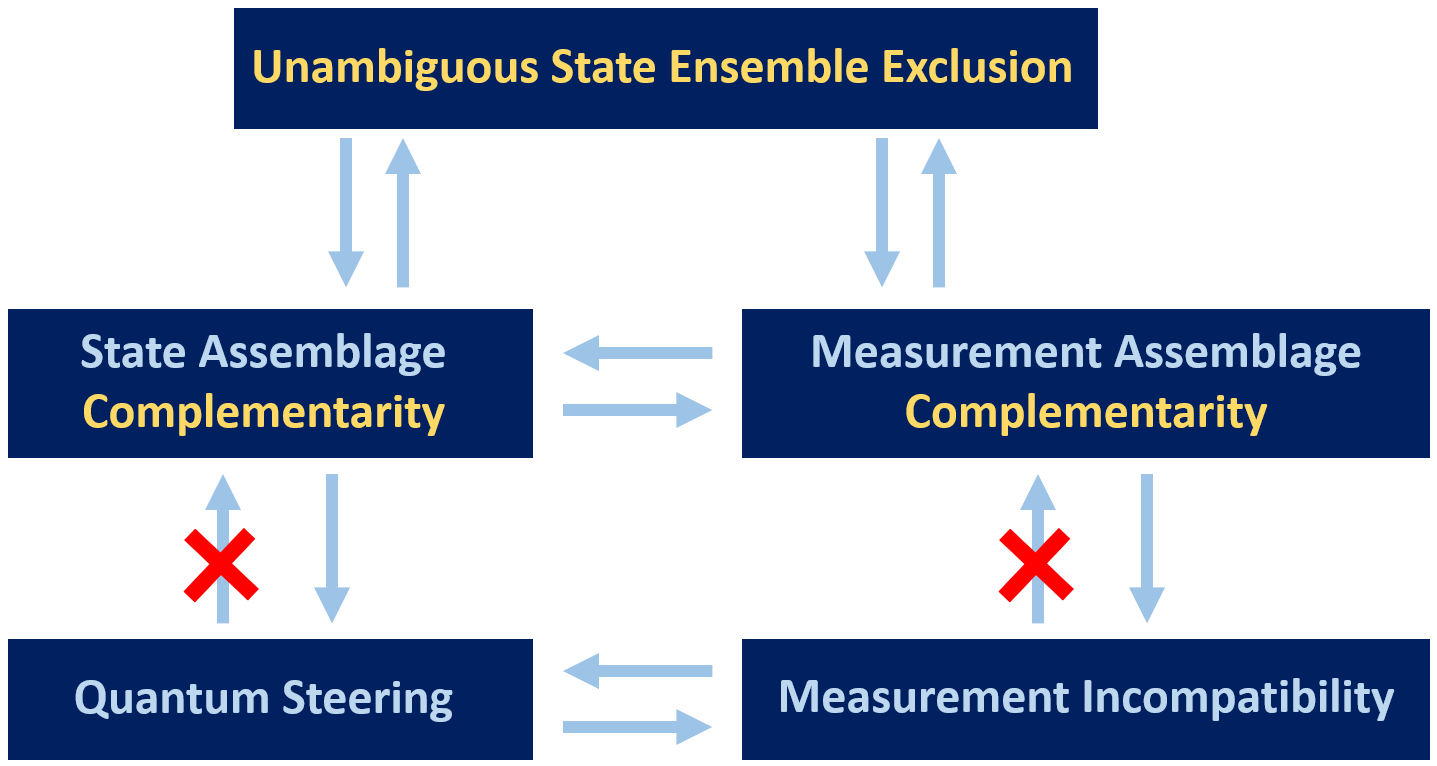}}
\caption{
{\bf A map of relation between quantum complementarity, unambiguous exclusion tasks, and other quantum resources.}
In this work, we show that quantum complementarity of state ensembles and measurements are both operationally equivalent to $\eta$-unambiguous exclusion, a novel task introduced here.
Moreover, our analysis shows that quantum complementarity can be viewed as a strong notion of measurement incompatibility as well as quantum steering.
This not only bridges quantum complementarity and known quantum resources, but also provide a novel operational way to study incompatibility and steering.
}
\label{Fig:ERE} 
\end{figure}

\section*{Discussions}
In this work, we have studied quantum complementarity and shown how it can be understood as a novel type of quantum resource.
In particular, we have given a complete analytical quantification of complementarity here. To do so, we introduce the novel class of $\eta$-unambiguous exclusion tasks, and show that they precisely allow us to quantify quantum complementarity. 
Our results have a number of important physical implications.

First, by viewing quantum complementarity as a strong notion of measurement incompatibility and quantum steering, our findings suggest that strong enough incompatibility and steering can provide advantages in some unambiguous exclusion tasks as well as encryption tasks.
Second, Theorem~\ref{Result:Exclusion-SE-maintext} gives a natural {\em one-sided device-independent} (1SDI) scenario that can certify ensemble complementarity.  
In the same direction, by combining this insight with Theorem~\ref{Result:ExclusionTask}, we can further conclude that quantum complementarity can be certified in an 1SDI way. 
This is consistent with the fact that certain types of quantum complementarity can be certified in a fully device-independent way, specifically those that can be self-tested~\cite{Supic2020}; for instance, spin measurements in $X,Z$ directions can be self-tested by Clauser-Horne-Shimony-Holt (CHSH) inequality~\cite{CHSH,Supic2020}.
Third, we remark that our results naturally apply to assemblages that are {\em not} complementary.
For instance, Theorem~\ref{Result:Encryption} implies that ${\bf E}$ can perform {\em almost} unambiguous entanglement-assisted $\eta$-encryption when ${\bf E}$ is almost complementary in this sense.
Finally, the tasks reported in this work are all feasible with current technology. 
Hence, our findings provide a recipe to experimentally certify quantum complementarity, which in turn would constitute experimental witness of the associated strong  measurement incompatibility and quantum steering.

Going forward, we first note that an alternative definition of measurement complementarity considers also outcomes such corresponding to \emph{coarse-grainings}, e.g., for the measurement outcome corresponding to $E_{1|0}+E_{2|0}$, and require $M=0$ when $E_{1|0}+E_{2|0}\ge M$ and $E_{b|1}\ge M$. 
As we discuss in a forthcoming companion paper~\cite{Hsieh-WIP}, similar results also can be proved when using this stricter definition of complementarity.
It would also be interesting to study $\eta$-unambiguous exclusion tasks in their own right, and the closely related \mbox{$\eta$-unambiguous} discrimination task.
Recent works showed that one possible way to extend exclusion-based tasks is using ideas from economics~\cite{Ducuara2023,Ducuara2023-2,Ducuara2022PRXQ}. 
It would be interesting to apply similar ideas to $\eta$-unambiguous tasks. 
Furthermore, following recent findings of steering distillation~\cite{Nery2020PRL,Liu2022,Ku2022NC,Ku2023,Hsieh2023}, it is useful to further study distillation of quantum complementarity, which could enhance quantum advantages in 1SDI protocols.
Finally, it would be rewarding to study the ability of quantum dynamics to {\em preserve} quantum complementarity~\cite{Hsieh2020,Hsieh2021PRXQ,Ku2022PRXQ} and its implications to information transmission.

\subsection*{Acknowledgements}
We thank Gelo Noel Tabia and Daniel Urrego for fruitful discussion and comments.
C.-Y.~H.~is supported by the Royal Society through Enhanced Research Expenses (on grant NFQI) and the ERC Advanced Grant (FLQuant). 
R.~U.~is thankful for the financial support from the Swiss National Science Foundation (Ambizione PZ00P2- 202179). 
P.~S.~is a CIFAR Azrieli Global Scholar in the Quantum Information Science Programme.

\section*{Appendix}
\subsection*{Properties of exclusive mutual information}
\begin{alemma}\label{Result:Properties}
Let $\stEns = \{\rho_x\}_x$ be a set of states.
Then we have
\begin{enumerate}
\item (Non-Negativity) $I_{\rm exc}(\stEns)\ge0$, and the equality holds if $\rho_x = \rho_y$ for every $x,y$.
\item (Data-Processing Inequality) For every channel $\mathcal{E}$, we have $I_{\rm exc}[\{\mathcal{E}(\rho_x)\}_x]\le I_{\rm exc}(\stEns)$.
\end{enumerate}
\end{alemma}
\begin{proof}
A feasible point $P$ of Eq.~\eqref{Eq:C-function} must satisfy \mbox{${\rm tr}(P)\le{\rm tr}(N_x)$} $\forall\;x$.
This means ${\rm tr}(P)\le1$ when we have quantum states, i.e.,  $I_{\rm exc}(\stEns)\ge0$.
When $\rho_x = \rho_y = \eta$ $\forall\;x,y$, we can choose $P=\eta$ to achieve $I_{\rm exc}(\stEns)=0$.
Finally, the data-processing inequality follows by a direct computation for every channel $\mathcal{E}$ (also known as {\em completely-positive trace-preserving} linear map~\cite{QIC-book}):
$
q_{\rm exc}[\{\mathcal{E}(\rho_x)\}_x]\ge\max\left\{{\rm tr}[\mathcal{E}(P)]\,|\,\mathcal{E}(\rho_x)\ge \mathcal{E}(P), P\ge0\right\}
\ge\max\left\{{\rm tr}(P)\,|\,\rho_x\ge P, P\ge0\right\}=q_{\rm exc}(\stEns).
$
Here, we change maximisation ranges from $\{P\,|\,P\ge0\}$ to $\{\mathcal{E}(P)\,|\,P\ge0\}$, and then use $\{P\ge0\,|\,\rho_x\ge P\}\subseteq\{P\ge0\,|\,\mathcal{E}(\rho_x)\ge\mathcal{E}(P)\}$.
\end{proof}
Lemma~\ref{Result:Properties} suggests that $I_{\rm exc}$ can be viewed as a mutual information.
As another important feature, $q_{\rm exc}$ is actually a semi-definite program~\cite{Watrous-book,SDP-textbook} with the following explicit dual form ($\Nset = \{N_x\}_x$ denotes a set of positive semi-definite operators):
\begin{align}\label{Eq:C-function-dual}
q_{\rm exc}(\Nset) = \min_{\substack{\sum_xQ_x\ge\id\\Q_x\ge0\;\;\forall\;x}}\sum_x{\rm tr}(Q_xN_x).
\end{align}
\begin{proof}
Consider the Lagrangian
$
\mathcal{L}(P,\{Q_x\}_x,R)\coloneqq{\rm tr}(P) + \sum_x{\rm tr}\left[Q_x\left(N_x-P\right)\right] + {\rm tr}(RP).
$
Then one can check that
$
q_{\rm exc}(\Nset) = \max_{P\ge0}\min_{\substack{Q_x\ge0\;\forall\,x\\R\ge0}}\mathcal{L}(P,\{Q_x\}_x,R).
$
Switching the order of minimisation and maximisation gives us the dual problem
$
\min_{\substack{Q_x\ge0\;\forall\,x\\R\ge0}}\max_{P\ge0}\mathcal{L}(P,\{Q_x\}_x,R),
$
where
\begin{align}\nonumber
\max_{P\ge0}\mathcal{L}(P,\{Q_x\}_x,R) = 
\begin{cases}
\sum_x{\rm tr}(Q_xN_x)&\text{if}\,\id+R\le \sum_xQ_x;\\
\infty&\text{otherwise}.
\end{cases}
\end{align}
From here we conclude that Eq.~\eqref{Eq:C-function-dual} is the dual.
Finally, one can check that the primal problem Eq.~\eqref{Eq:C-function} is finite and feasible (e.g., by choosing $P=0$), and the above dual problem is strictly feasible; namely, there exist $Q_x>0$ with $\sum_xQ_x>\id$ (for example, this can be achieved by choosing $Q_x=\id$ for every $x$).
Hence, by Slater's condition (e.g., Theorem 1.18 in Ref.~\cite{Watrous-book}), the strong duality holds, and the primal and the dual output the same optimal value.
\end{proof}

Now we report the following representation of $q_{\rm exc}$, which is the major tool for us to explore various exclusion tasks.
\begin{alemma}\label{lemma:characterisation}
There is $\eta_*<1$ such that, for every \mbox{$\eta_*\le\eta<1$,}
\begin{align}\label{Eq:D-function}
(1-\eta)q_{\rm exc}\left(\Nset\right)=D_\eta(\Nset)
\coloneqq\min_{\substack{
(1-\eta)\id\le \sum_xQ_x\le\id\\Q_x\ge0\;\;\forall\;x}}\sum_x{\rm tr}\left(Q_xN_x\right).
\end{align}
\end{alemma}
\begin{proof}
Write $q_{\rm exc}\left(\Nset\right) = \sum_x{\rm tr}\left(Q_{x}^{\rm opt}N_x\right)$ with $\id\le\sum_xQ_{x}^{\rm opt}$, where $Q_x^{\rm opt}$'s are optimal operators achieving Eq.~\eqref{Eq:C-function-dual}.
Let 
\begin{align}\label{Eq:epsilon_*}
\eta_*\coloneqq1-\frac{1}{\norm{\sum_xQ_{x}^{\rm opt}}_\infty},
\end{align}
which satisfies $0\le\eta_*<1$.
Then for every $\eta_*\le\eta<1$ and every $x$, consider
$
Q_x^{(\eta)}\coloneqq(1-\eta)Q_x^{\rm opt}.
$
One can check that
$
(1-\eta)\id\le\sum_xQ_x^{(\eta)}\le(1-\eta_*)\sum_xQ_{x}^{\rm opt}\le\id,
$
meaning that $\{Q_x^{(\eta)}\}_x$ is a feasible point of the minimisation in  Eq.~\eqref{Eq:D-function}.
Hence,
$
D_\eta(\Nset)\le\sum_x{\rm tr}\left(Q_x^{(\eta)}N_x\right)
=(1-\eta)q_{\rm exc}\left(\Nset\right).
$
On the other hand,
\begin{align}
(1-\eta)q_{\rm exc}\left(\Nset\right)&\le(1-\eta)\min_{\substack{\id\le \sum_xQ_x\le\frac{1}{1-\eta}\id\\Q_x\ge0\;\forall\,x}}\sum_x{\rm tr}\left(Q_xN_x\right)\nonumber\\
&=\min_{\substack{(1-\eta)\id\le\sum_xQ_x\le\id\\Q_x\ge0\;\forall\,x}}\sum_x{\rm tr}\left(Q_xN_x\right)=D_\eta(\Nset),\nonumber
\end{align}
where the first inequality is the consequence of reducing the minimisation range.
The result thus follows.
\end{proof}

\subsection*{Proof of Theorem~\ref{Result:ERE-Task-maintext}}
\begin{proof}
By Lemma~\ref{lemma:characterisation}, there is $\eta_*<1$ such that,
\begin{align}
q_{\rm exc}(\stEns)
&=\min_{\prob>0}\frac{\min_{\substack{(1-\eta)\id\le \sum_xQ_x\le\id\\Q_x\ge0\;\forall\,x}}\sum_x\left(\min_yp_y\right){\rm tr}(Q_x\rho_x)}{\left(\min_yp_y\right)(1-\eta)}\nonumber\\
&\le\min_{\prob>0}\frac{\min_{\substack{(1-\eta)\id\le \sum_xQ_x\le\id\\Q_x\ge0\;\forall\,x}}\sum_xp_x{\rm tr}(Q_x\rho_x)}{\Perrorclass\left(\prob,\eta\right)}\nonumber
\end{align}
for every $\eta_*\le\eta<1$.
Using Lemma~\ref{lemma:characterisation} again, we conclude that this bound can be saturated by $\{p_x = 1/\cardin\}_x$.
\end{proof}
Finally, we remark that Theorem~\ref{Result:ERE-Task-maintext} also holds when we replace $\stEns$ by $\Nset$, a set of positive semi-definite operators.

\subsection*{Proof of Theorem~\ref{Result:Exclusion-SE-maintext}}
\begin{proof}
First, it is useful to note that we can write
\begin{align}\label{Eq:useful}
\Perrorens\left({\bm\sigma},\prob,\eta\right)=\max_{\permu}\sum_{a=0}^{\cardout-1}D_\eta(\{p_x\sigma_{\pi_x(a)|x}\}_x),
\end{align}
where $D_\eta$ is defined in Eq.~\eqref{Eq:D-function}.
By Lemma~\ref{lemma:characterisation}, for every index set $\{\pi_x(a)\}_x$, there exists \mbox{$\eta_{*}(\{\pi_x(a)\}_x)<1$} such that
\mbox{$
(1-\eta)q_{\rm exc}(\{\sigma_{\pi_x(a)|x}\}_x)=D_\eta(\{\sigma_{\pi_x(a)|x}\}_x)
$}
for every $\eta_{*}(\{\pi_x(a)\}_x)\le\eta<1$.
Let 
$
\eta_*\coloneqq\max_{\permu,a}\eta_{*}(\{\pi_x(a)\}_x).
$
Using Lemma~\ref{lemma:characterisation}, we have ($\probunif\coloneqq\{p_x = 1/\cardin\}_x$)
\begin{align}\label{Eq:CSE-relation}
&\frac{\Perrorens\left({\bm\sigma},\probunif,\eta\right)}{\Perrorclass(\probunif,\eta)}=\max_{\permu}\sum_{a=0}^{\cardout-1}\frac{D_\eta(\{\sigma_{\pi_x(a)|x}\}_x)}{1-\eta}\nonumber\\
&=\max_{\permu}\sum_{a=0}^{\cardout-1}q_{\rm exc}(\{\sigma_{\pi_x(a)|x}\}_x)=C_{\rm SE}({\bm\sigma})
\end{align}
for every \mbox{$\eta_*\le\eta<1$,}
this proves the inequality `$\ge$' in Theorem~\ref{Result:Exclusion-SE-maintext}.
To see the inequality `$\le$', Eq.~\eqref{Eq:CSE-relation} implies
\begin{align}
&C_{\rm SE}({\bm\sigma}) = \max_{\permu}\sum_{a=0}^{\cardout-1}\min_{\substack{(1-\eta)\id\le\sum_xQ_x\le\id\\Q_x\ge0\;\;\forall\;x}}\sum_{x=0}^{\cardin-1}\frac{{\rm tr}\left(Q_x\sigma_{\pi_x(a)|x}\right)}{1-\eta}\nonumber\\
& =  \max_{\permu}\sum_{a=0}^{\cardout-1}\min_{\substack{(1-\eta)\id\le\sum_xQ_x\le\id\\Q_x\ge0\;\;\forall\;x}}\sum_{x=0}^{\cardin-1}\frac{(\min_yp_y){\rm tr}\left(Q_x\sigma_{\pi_x(a)|x}\right)}{\Perrorclass\left(\prob,\eta\right)}\nonumber\\
& \le 
\frac{\Perrorens({\bm\sigma},\prob,\eta)}{\Perrorstate \left(\{\averho\}_x,\prob,\eta\right)}\nonumber
\end{align}
for every $\prob>0$ and $\eta_*\le\eta<1$;
note that for non-signalling state assemblage ${\bm\sigma}$ we have $\Perrorstate \left(\{\averho\}_x,\prob,\eta\right) = \Perrorclass(\prob,\eta)$.
Finally, taking minimisation over $\prob>0$ proves the desired claim.
\end{proof}
For a measurement assemblage ${\bf E}$, the set $\{E_{a|x}/d\}_{a,x}$ is a state assemblage, and $C_{\rm POVM}({\bf E})/d = C_{\rm SE}(\{E_{a|x}/d\}_{a,x})$. Hence, using Eq.~\eqref{Eq:CSE-relation}, we obtain
\begin{acorollary}\label{CDLemma}
Let ${\bf E}$ be a measurement assemblage.
Then there exists $\eta_*<1$ such that, for every $\eta_*\le\eta<1$,
\begin{align}
(1-\eta)C_{\rm POVM}({\bf E})=\max_{\permu}\sum_{a=0}^{\cardout-1}D_\eta(\{E_{\pi_x(a)|x}\}_x).
\end{align}
\end{acorollary}

\subsection*{No-go results of $\eta$-unambiguous exclusion tasks}
As mentioned in the main text, a natural question is whether one can improve the performance in $\eta$-unambiguous state-exclusion tasks if we also allow the player to do more.
The allow operations include (i) being able to make use of certain pre-measurement quantum channels $\mathcal{E}_\lambda$, each being implemented with probability $q_\lambda$, and (ii) applying post-processing of the output classical data based on another probability distribution $\widetilde{P}(x|y,\lambda)$ with $x,y=0,...,\cardin-1$.
In this setting, we further require 
\begin{align}\label{Eq:additional condition}
\sum_{x=0}^{\cardin-1}\widetilde{P}(x|y,\lambda) = 1\quad\forall\;y=0,...,\cardin-1.
\end{align}
Hence, one cannot use post-processing to increase the probability (i.e., generate the chance of inconclusion) of the outcome $\emptyset$ and artificially lower the error probability.
Equation~\eqref{Eq:additional condition} can thus be viewed as the {\em inconclusiveness-non-generating} condition.
We collectively write \mbox{$\Theta\coloneqq\left(\{q_\lambda\}_\lambda,\{\widetilde{P}(k|y,\lambda)\}_{k,y,\lambda},\{\mathcal{E}_\lambda\}_\lambda\right)$} to denote one combination of allowed operations.
For a given $\Theta$ and a set of states $\stEns = \{\rho_x\}_x$, the smallest error probability reads
\begin{align}
&\PerrorstateTheta\left(\stEns,\prob,\eta\right) \coloneqq \nonumber\\
&\min_{\substack{(1-\epsilon)\id\le\sum_xQ_x\le\id\\Q_x\ge0\;\;\forall\;x}}\sum_{x,y=0}^{\cardin-1}\sum_\lambda p_xq_\lambda\widetilde{P}(x|y,\lambda){\rm tr}[Q_y\mathcal{E}_\lambda(\rho_x)].
\end{align}
Then we have the following {\em no-go result} --- non-trivial pre-measurement channels and post-processing {\em cannot} provide any improvement in $\eta$-unambiguous state-exclusion tasks, as long as condition in Eq.~\eqref{Eq:additional condition} holds:
\begin{atheorem}\label{Result:ERE-Task}
For every set of states $\stEns$, there exists $\eta_*<1$ such that
\begin{align}\label{Eq:allowed op no-go}
\Perrorstate \left(\stEns,\prob,\eta\right)\le\PerrorstateTheta\left(\stEns,\prob,\eta\right)
\end{align}
for every $\prob>0$, \mbox{$\eta_*\le\eta<1$}, and $\Theta$.
Consequently,
\begin{align}\label{Eq:Result:ERE-Task2}
q_{\rm exc}\left(\stEns\right)
= \min_{\prob>0,\Theta}\frac{\PerrorstateTheta\left(\stEns,\prob,\eta\right)}{\Perrorclass\left(\prob,\eta\right)}\quad\forall\,\eta_*\le\eta<1.
\end{align}
\end{atheorem}
\begin{proof}
It suffices to prove Eq.~\eqref{Eq:allowed op no-go}, since Eq.~\eqref{Eq:Result:ERE-Task2} is its direct consequence.
Consider a given collection of allowed operations $\Theta$.
For every $\prob>0$ and $\eta_*\le\eta<1$ (where $\eta_*$ is given by Lemma~\ref{lemma:characterisation}), we have
\begin{align}
\PerrorstateTheta\left(\stEns,\prob,\eta\right)=\min_{\substack{(1-\eta)\id\le\sum_xQ_x\le\id\\Q_x\ge0\;\;\forall\;x}}\sum_{x=0}^{\cardin-1}p_x{\rm tr}\left(L_x^{(\Theta)}\rho_x\right),\nonumber
\end{align}
where, for every $x=0,..,\cardin-1$,
\begin{align}\label{Eq:L-func}
L_x^{(\Theta)}\coloneqq\sum_{y=0}^{\cardin-1}\sum_\lambda q_\lambda\widetilde{P}(x|y,\lambda)\mathcal{E}_\lambda^\dagger(Q_y)\ge0.
\end{align}
Using Eq.~\eqref{Eq:additional condition}, we have
$
(1-\eta)\id\le\sum_{x=0}^{\cardin-1}L_x^{(\Theta)} = \sum_\lambda q_\lambda\mathcal{E}_\lambda^\dagger\left(\sum_{y=0}^{\cardin-1}Q_y\right)\le\sum_\lambda q_\lambda\mathcal{E}_\lambda^\dagger(\id) = \id.
$
Consequently, $\{L_x^{(\Theta)}\}_{x=0}^{\cardin-1}$ is a feasible point of Eq.~\eqref{Eq:P-error}; i.e., $\Perrorstate \left(\stEns,\prob,\eta\right)$.
This implies Eq.~\eqref{Eq:allowed op no-go}.
\end{proof}
As a direct corollary, inconclusiveness-non-generating allowed operations can neither improve state-ensemble exclusion.
To be precise, for a given $\Theta$ and a state assemblage ${\bm\sigma}$, the smallest error probability can be written as
\begin{multline}\label{Eq:P-error2}
\PerrorensTheta\left({\bm\sigma},\prob,\eta\right) \coloneqq\\
\max_{\permu}\sum_m p'_m \PerrorstateTheta\left(\{\rho_{\pi_x^{-1}(m)|x}\}_x,\{p'_{x|m}\}_x,\eta\right),
\end{multline}
where $\{\rho_{\pi_x^{-1}(m)|x}\}_x$, $p'_{x|m}$, and $p'_m$ are defined below Eq.~\eqref{Eq:p-ensemble}.
Using Eq.~\eqref{Eq:allowed op no-go} and the definition of $\Perrorens$, it is straightforward to see that, for every non-signalling state assemblage ${\bm\sigma}$,
there exists $\eta_*<1$ such that
\begin{align}
\Perrorens\left({\bm\sigma},\prob,\eta\right)\le\PerrorensTheta\left({\bm\sigma},\prob,\eta\right)\end{align}
 for every $\prob>0$, $\eta_*\le\eta<1$, and $\Theta$.
Hence, we again have
\begin{align}
C_{\rm SE}({\bm\sigma}) = \min_{\prob>0,\Theta}\frac{\PerrorensTheta\left({\bm\sigma},\prob,\eta\right)}{\Perrorclass\left(\prob,\eta\right)}.
\end{align}

\subsection*{Proof of Theorem~\ref{Result:ExclusionTask}}
\begin{proof} 
Consider a given $\rho_{AB}$ with full-rank marginals.
By Eq.~\eqref{Eq:CSE-relation} and Lemma~\ref{lemma:characterisation}, there exists \mbox{$\eta_{(1)}<1$} such that, for every $\eta_{(1)}\le\eta<1$, we have
(again, $\probunif\coloneqq\{p_x = 1/\cardin\}_x$)
\begin{align}\label{Eq:Comp_001}
&C_{\rm SE}\left[{\bm\sigma}({\bf E},\rho_{AB})\right]=\frac{\cardin}{1-\eta}\times\Perrorens\left[{\bm\sigma}({\bf E},\rho_{AB}),\probunif,\eta\right]\nonumber\\
& = \max_{\permu}\sum_{a=0}^{\cardout-1}\min_{\substack{(1-\eta)\id\le \sum_xQ_x\le\id\\ Q_x\ge0}}\sum_{x=0}^{\cardin-1}\frac{{\rm tr}\left(E_{\pi_x(a)|x}\tau_x\right)}{1-\eta}.
\end{align}
Here, we let
$ 
\tau_x\coloneqq{\rm tr}_B\left[(\id_A\otimes Q_x)\rho_{AB}\right]
$,
which satisfies
$
(1~-~\eta)~\mu_{\rm min}~(\rho_{A})\id_A\le(1-\eta)\rho_A\le\sum_{x=0}^{\cardin-1}\tau_x \le \rho_A\le\id_A.
$
Define
\mbox{$
\delta_\eta\coloneqq1-(1-\eta)\mu_{\rm min}(\rho_A).
$}
Using Eqs.~\eqref{Eq:Comp_001} and~\eqref{Eq:D-function}, we have, for every \mbox{$\eta_{(1)}\le\eta<1$,}
\begin{align}
(1-\eta)C_{\rm SE}\left[{\bm\sigma}({\bf E},\rho_{AB})\right]\ge\max_{\permu}\sum_{a=0}^{\cardout-1}D_{\delta_\eta}(\{E_{\pi_x(a)|x}\}_x).\nonumber
\end{align}
Now, by Corollary~\ref{CDLemma}, there exists \mbox{$\eta_{(2)}<1$} such that, for every $\eta_{(2)}\le\eta<1$,
\begin{align}\label{Eq:Comp-000}
(1-\eta)C_{\rm POVM}({\bf E})=\max_{\permu}\sum_{a=0}^{\cardout-1}D_\eta(\{E_{\pi_x(a)|x}\}_x).
\end{align}
Finally, when $\eta\ge1-\left(1-\eta_{(2)}\right)/\mu_{\rm min}(\rho_A)$, we have \mbox{$\delta_\eta\ge\eta_*$,} meaning that Eq.~\eqref{Eq:Comp-000} is working for $\delta_\eta$.
Altogether, for every $\max\left\{\eta_{(1)},\eta_{(2)},1-\left(1-\eta_{(2)}\right)/\mu_{\rm min}(\rho_A)\right\}\le\eta<1$, we have
$
(1-\eta)C_{\rm SE}\left[{\bm\sigma}({\bf E},\rho_{AB})\right]\ge(1-\delta_\eta)C_{\rm POVM}({\bf E})
$, which implies the desired upper bound.

To saturate the upper bound, consider $\rho_{AB} = \phi^+_{AB}\coloneqq\proj{\phi^+_{AB}}$.
Using Eq.~\eqref{Eq:CSE-relation} and Corollary~\ref{CDLemma}, there is $\eta_*<1$ such that, for every \mbox{$\eta_*\le\eta<1$,}
\begin{align}
&C_{\rm SE}\left[{\bm\sigma}({\bf E},\phi^+_{AB})\right]=\nonumber\\
&\max_{\permu}\sum_{a=0}^{\cardout-1}\min_{\substack{(1-\eta)\id\le \sum_xQ_x\le\id\\ Q_x\ge0}}\sum_x\frac{{\rm tr}\left[\left(E_{\pi_x(a)|x}Q_x^\intercal\otimes\id_B \right)\phi^+_{AB}\right]}{1-\eta}\nonumber\\
&=\max_{\permu}\sum_{a=0}^{\cardout-1}\min_{\substack{(1-\eta)\id\le \sum_xQ_x\le\id\\ Q_x\ge0}}\frac{1}{d}\sum_x\frac{{\rm tr}\left(E_{\pi_x(a)|x}Q_x\right)}{1-\eta}\nonumber\\
&=\mu_{\rm min}(\phi^+_{AB})C_{\rm POVM}({\bf E}).\nonumber
\end{align}
Here, $(\cdot)^\intercal$ is the transpose map, and we have used the fact that $(\id\otimes M)\ket{\phi^+_{AB}} = (M^\intercal\otimes\id)\ket{\phi^+_{AB}}$ for every operator $M$ and $\mu_{\rm min}(\phi^+_{AB}) = 1/d$.
\end{proof}

\subsection*{Proof of Theorem~\ref{Result:Encryption}}
\begin{proof}
Using Corollary~\ref{CDLemma}, there exists $\eta_*<1$ such that
\begin{align}
\Perrorencrypt({\bf E},\eta)& = \max_{\permu}\sum_{m=0}^{\cardout-1}\min_{\substack{(1-\eta)\id\le Q_0+Q_1\le\id\\ Q_0\ge0,Q_1\ge0}}\sum_{x=0,1}\frac{{\rm tr}\left(E_{\pi_x(a)|x} Q_x\right)}{2d}\nonumber\\
&=\frac{1-\eta}{2d}C_{\rm POVM}({\bf E})\quad\forall\,\eta_*\le\eta<1,\label{Eq:Computation01}
\end{align}
where we have used Eq.~\eqref{Eq:p-error} in the first line.
\end{proof}

\subsection*{Proof of Eq.~\eqref{Eq:Result:No-Go}: A no-go result for deterministic unambiguous $\eta$-encryption tasks}
\begin{proof}
Suppose $\Perrorencrypt({\bf E},\eta=0) = 0.$
Let us pick one index $a_*$ satisfying $E_{a_*|0}\neq0$.
Then Eq.~\eqref{Eq:Computation01} implies that, for every $b$, there exists $0\le P_{b}\le\id$ such that
$
{\rm tr}\left(E_{a_*|0}P_{b}\right) + {\rm tr}\left[E_{b|1}(\id-P_{b})\right] = 0.
$
Thus,
$
{\rm tr}\left(\eta_{b}P_{b}\right) ={ \rm tr}\left[\rho_{a_*}(\id-P_{b})\right] = 1$ $\forall\;b\in\mathcal{I}_+\coloneqq\{b\;|\;E_{b|1}\neq0\},
$
where 
$\rho_{a_*}\coloneqq E_{a_*|0}/{\rm tr}\left(E_{a_*|0}\right)$ and $\eta_{b}\coloneqq E_{b|1}/{\rm tr}\left(E_{b|1}\right)$ are quantum states.
From the operational interpretation of trace distance, we recall that
$
2 = \max_{0\le P\le\id}\left[{\rm tr}(\rho P) + {\rm tr}\left(\eta(\id-P)\right)\right]
$
if and only if two states $\rho,\eta$ have orthogonal supports.
This means
\mbox{$
{\rm supp}\left(E_{a_*|0}\right)\perp{\rm supp}\left(E_{b|1}\right)$} $\forall\;b\in\mathcal{I}_+;
$
i.e.,
$
E_{b|1}\ket{\psi} = 0$ $\forall\;b\in\mathcal{I}_+ $ and $\forall\;\ket{\psi}\in{\rm supp}\left(E_{a_*|0}\right).$
From here we observe that
$\ket{\psi} = \id\ket{\psi} = \sum_{b\in\mathcal{I}_+} E_{b|1}\ket{\psi} = 0$ $\forall\;\ket{\psi}\in{\rm supp}\left(E_{a_*|0}\right),$
which implies $E_{a_*|0} = 0$, a contradiction.
\end{proof}

\subsection*{Computation details for the example}
Consider a pair of two-outcome measurements given by
$
E_{0|0} = \proj{0}, E_{1|0}=\proj{1}, E_{0|1} = \proj{+}, E_{1|1}=\proj{-}.
$
Using Eq.~\eqref{Eq:Computation01}, zero error probability with the encrypt strategy $\permu = \{\pi_0,\pi_1\}$ can be achieved by choosing
\begin{align}
&Q_{x_{\rm out}=1|\pi_{x=0}(a=0)}^\intercal=\alpha\proj{1}; Q_{x_{\rm out}=1|\pi_{x=0}(a=1)}^\intercal=\alpha\proj{0};\nonumber\\
&Q_{x_{\rm out}=0|\pi_{x=1}(a=0)}^\intercal=\alpha\proj{-}; Q_{x_{\rm out}=0|\pi_{x=1}(a=1)}^\intercal=\alpha\proj{+},\nonumber
\end{align}
where, for simplicity, we consider a single proportional constant $\alpha$.
Now we analyse the success probability.
From the definition of inconclusiveness $\eta$, we have
$
(1-\eta)\id\le\alpha\left(\proj{0}+\proj{+}\right)\le\id
$.
This implies $
\alpha\left(1+1/\sqrt{2}\right)\le1$ and $1-\eta\le\alpha\left(1-1/\sqrt{2}\right).
$
Note that replacing $\proj{0}$ by $\proj{1}$ and/or replacing $\proj{+}$ by $\proj{-}$ gives the same constraints on $\alpha$, so it suffices to analyse the above case.
The highest value of $\alpha$ reads
$
\alpha_{\rm max} = \sqrt{2}/\left(1+\sqrt{2}\right).
$
This means that zero error probability can be achieved for every $\eta_{\rm min}\le\eta<1$, where
$
\eta_{\rm min} = 1-\alpha_{\rm max}\left(1-1/\sqrt{2}\right) = 2/\left(1+\sqrt{2}\right).
$
Finally, when we choose $\eta=\eta_{\rm min}$ and $\alpha = \alpha_{\rm max}$, one gets the decoding strategy given in the main text. 
The success probability 
reads $\sum_{m,x=0,1}\frac{1}{2}\bra{{\phi_{AB}^+}}\left(E_{\pi^{-1}_{x}(m)|x}\otimes Q_{x|m}\right)\ket{\phi_{AB}^+} = \alpha_{\rm max}/2\approx0.2929.$

\subsection*{Proof of Theorem~\ref{Result:IncompatibilityWeight}}
\begin{proof}
A compatible measurement assemblage \mbox{${\bf L} = \{L_{a|x}\}_{a,x}$} can be written as $L_{a|x} = \sum_{i}D(a|x,i)G_i$ $\forall\;a,x$,
where $\{G_i\}_i$ is a POVM, and each $D(a|x,i)$ is a deterministic probability distribution; namely, it assigns each input $x$ with precisely one outcome $a$~\cite{SDP-textbook}.
Hence, we can write 
$
1-{\rm IW}({\bf E})=\max\{0\le q\le1\;|\;E_{a|x}\ge q \sum_i D(a|x,i) G_i\;\forall\;a,x; \{G_i\}_i\text{: POVM}\}.
$
By considering the variable $K_i\coloneqq qG_i$ and extending the maximisation range, $1-{\rm IW}({\bf E})$ is upper bounded by
\begin{eqnarray}
\begin{aligned}
\max\quad&\frac{1}{d}{\rm tr}\left(\sum_i K_i\right)\\
{\rm s.t.}\quad&E_{a|x}\ge\sum_i D(a|x,i) K_i\;\;\forall\;a,x; \\
&K_i\ge0; \sum_i K_i=\frac{\id}{d}{\rm tr}\left(\sum_i K_i\right).\nonumber
\end{aligned}
\end{eqnarray}
Write ${\bm b} = (b_0,b_1,...,b_{\cardin-1})$, and use $i_{\bm b}$ denote the deterministic probability distribution that maps $x$ to $b_x$ for every $x=0,...,\cardin-1$; that is, $D(a=b_x|x,i_{\bm b})=1$ $\forall\;x$.
Then
$
\sum_i {\rm tr}\left(K_i\right) = \sum_{\bm b} {\rm tr}\left(K_{i_{\bm b}}\right),
$
and $1-{\rm IW}({\bf E})$ is upper bounded by
\begin{align}\nonumber
\sum_{\bm b}\max\left\{{\rm tr}\left(K_{i_{\bm b}}\right)/d\;\middle|\;E_{b_x|x}\ge\sum_i D(b_x|x,i) K_i\;\forall\,x;\; K_i\ge0\right\}
\end{align}
where we again extend the maximisation range.
Finally, since
$
\sum_i D(b_x|x,i) K_i=\sum_{i: D(b_x|x,i) = 1}K_i\ge K_{i_{\bm b}}
$
and $\sum_{\bm b} = \cardout^\cardin$, the above maximisation is upper bounded by $\cardout^\cardin C_{\rm POVM}/d$.
\end{proof}

\bibliography{Ref.bib}

\end{document}